\documentclass[conference,10pt]{IEEEtran}
\usepackage[top=0.75in, bottom=1.1 in, left=0.64in, right=0.64in]{geometry}
\usepackage[T1]{fontenc}
\usepackage{amsmath,amssymb,cite,booktabs}
\usepackage{comment}
\usepackage{balance}
\usepackage{setspace}
\usepackage[hidelinks]{hyperref}
\newtheorem{definition}{Definition}
\newtheorem{lemma}{Lemma}
\newtheorem{theorem}{Theorem}
\newcommand{\BB}{\mathbb{B}}
\newcommand{\FF}{\mathbb{F}}
\newcommand{\cC}{\mathcal{C}}

\newcommand{\partyking}{P_{\mathsf{king}}}

\newcommand{\Dec}{\operatorname{Dec}}

\newcommand{\Mpoly}{\mathsf{M}}
\makeatletter
\def\@IEEEabskeysecsize{\normalsize}
\let\originalbibliography\thebibliography
\renewcommand{\thebibliography}[1]{\originalbibliography{#1}\normalsize}
\makeatother
\title{Big Secrets Small Shares: Revisiting Secure Multi-Party Computation}
\author{\IEEEauthorblockN{Wilton Kim\IEEEauthorrefmark{1}, Van Long Phuoc Pham\IEEEauthorrefmark{2},
Stanislav Kruglik\IEEEauthorrefmark{3}, Han Mao Kiah\IEEEauthorrefmark{1}}
\IEEEauthorblockA{\IEEEauthorrefmark{1}Nanyang Technological University, Singapore
\quad \texttt{wilton.kim@ntu.edu.sg, hmkiah@ntu.edu.sg}\\
\IEEEauthorrefmark{2}Brown University, USA \quad \texttt{phuoc\_pham\_van\_long@brown.edu}\\
\IEEEauthorrefmark{3}Technical University of Denmark, Denmark \quad \texttt{stakr@dtu.dk}}}
\begin{document}
\maketitle
\thispagestyle{empty}
\begin{abstract}
At CRYPTO 2007, Damg{\aa}rd and Nielsen introduced a preprocessing approach to multi-party computation (MPC) to achieve linear communication complexity. Another approach to reduce communication complexity was proposed by Cramer, Damg{\aa}rd, and de Haan at EUROCRYPT 2007, which involved using ramp-secret sharing schemes to jointly share base-field components of secrets from the extension field. In this paper, we combine these two ideas to achieve a non-asymptotic reduction of communication cost in a multi-phase MPC setup. We introduce a new modification of the ramp secret sharing scheme, called compact sharing, and tailor the MPC protocol to accommodate this modification. Compact sharing masks the polynomial coefficients that would otherwise leak during degree reduction. The protocol provides perfect security against static semi-honest adversaries under an explicit tradeoff between the corruption threshold and the extension degree. When the extension degree scales linearly with the number of parties, the gate-evaluation core achieves constant communication overhead relative to an extension-field element, amortized over full preprocessing batches; input and output communication is counted separately.
\end{abstract}

\section{Introduction}
Secure Multi-Party Computation (MPC~\cite{yao1982protocols}) enables a set of parties to jointly compute a function of their private inputs without revealing their inputs to each other. Formally, given $n$ parties $P_1,\ldots,P_n$, each holding a private input $x_i$, the goal is to securely compute a function $f(x_1,\ldots,x_n)$ while revealing only what is implied by the prescribed output and the corrupted parties' inputs. Information-theoretic MPC was established by Ben-Or, Goldwasser, and Wigderson (BGW)~\cite{wigderson1988completeness} and independently by Chaum, Cr{\'e}peau, and Damg{\aa}rd~\cite{chaum1988multiparty}, with secret sharing as a fundamental building block.

In this work, we construct an information-theoretic MPC protocol against semi-honest adversaries with a strong honest majority. Our protocol focuses on computation over extension fields. Specifically, for an extension field $\FF$ of degree $e$ over a base field $\BB$, we introduce an MPC protocol that uses shares from only the small field $\BB$. Each share therefore has $1/e$ of the size of an extension-field element. Such a reduction in share size lowers the overall MPC communication bandwidth at the cost of a reduced corruption threshold.

\subsection{Related Work and Our Contribution}

\paragraph{Information-theoretic MPC}
Given secrets in a finite field $\mathbb{F}$, the BGW~\cite{wigderson1988completeness} protocol achieves communication complexity that is quadratic in the number of parties, i.e., $O(n^2 \log |\mathbb{F}|)$ bits of communication per multiplication gate. Later, a breakthrough by Damg\aa rd and Nielsen~\cite{damgaard2007scalable} introduced a preprocessing approach that dramatically improved this dependency to linear, more specifically, $O(n \log |\mathbb{F}|)$ bits of communication per gate. Around the same time, Cramer, Damg\aa rd, and de Haan~\cite{cramer2007atomic} presented a scheme known as atomic MPC that improves the BGW~\cite{wigderson1988completeness} protocol in a different direction, namely achieving $O(n^2 \log |\mathbb{B}|)$ bits per gate when $\mathbb{F}$ is an extension field of $\mathbb{B}$. We combine these approaches by introducing compact sharing to prevent leakage during degree reduction. Under the stated parameter constraints, and amortizing over full preprocessing batches, the protocol communicates $O(n\log|\BB|)=O((n/e)\log|\FF|)$ bits per generalized multiplication gate, excluding direct input sharing and output reconstruction and dissemination. Its communication overhead relative to one extension-field element is therefore $O(n/e)$, becoming constant when $e=\Theta(n)$.

\paragraph{MPC with small communication overhead}
Several recent works further advanced the communication efficiency of secure MPC. Goyal et al.~\cite{goyal2020} and subsequently ATLAS~\cite{goyal2021atlas} optimized the preprocessing and online phases, respectively, achieving constant-factor improvements. Other notable protocols by Goyal, Polychroniadou, and Song~\cite{goyal2021hall}, and subsequent extensions~\cite{goyal2022transformation}, applied packed secret sharing~\cite{franklin1992communication}, bundling $k$ secrets together, thereby reducing per-secret communication by a factor of $k$ and achieving $O\big(n \frac{\log |\mathbb{F}|}{k}\big)$ bits of communication per gate. However, such sharing transformation approaches introduce complexities around the ordering and unpacking of secrets that incur additional communication overhead, and also need an $O(\mathrm{poly}(n))$ computation overhead per multiplication, in contrast to the $O(n \log^2 n)$ total computation ($O(\log^2 n)$ per party) in this paper. Another recent line of work based on unpacked secret sharing~\cite{garg2024scalable} achieves constant communication overhead for rings $\mathbb{Z}_p$. However, the current state-of-the-art protocols in this line of work require the ring bit size to be quadratic in the number of parties, i.e., $\log p = O(n^2 \log n)$, which is impractical in many settings. Hegde, Pham, and Wang subsequently improved its large-modulus requirement~\cite{hegde2026minimal}. These results concern prime fields and integer residue rings, whereas our construction directly represents arithmetic in a specified finite extension field through base-field polynomial coefficients.

\paragraph{Related secret sharing schemes}
In the context of bandwidth-efficient secret sharing, some previous works~\cite{huang2017secret,ding2021communication} have constructed secret sharing schemes for extension fields $\mathbb{F}$ for which secret reconstruction uses only elements from the base field $\mathbb{B}$. Specifically,~\cite{huang2017secret} proposed secret sharing schemes for extension-field elements utilizing both array codes and trace-repair techniques. Building on this,~\cite{ding2021communication} encodes secrets in the extension field similarly to a standard Shamir secret sharing scheme. To reduce secret reconstruction bandwidth, the scheme introduces a share transformation step, requiring parties to locally apply a trace function to their shares to project them into the base field $\mathbb{B}$ (or a subfield). However, an MPC protocol from these schemes has not yet been constructed.

\section{Preliminaries and Atomic Sharing}\label{sec:prelim}

\subsection{Fields, Parties, and Security Model}
We fix a finite extension field $\FF = \BB[X]/(F(X))$, where the base field $\BB$ has size $q$ and $F(X)$ is a monic irreducible polynomial of degree $e \ge 2$. Let $\omega \in \FF$ be a root of $F(X)$. Each element $a \in \FF$ can be uniquely written as a $\BB$-linear combination of powers of $\omega$:
\begin{equation}
  a = \sum_{j=0}^{e-1} a_j \omega^j, \qquad
  \tilde{a}(X) = \sum_{j=0}^{e-1} a_j X^j, \qquad a_j \in \BB.
\end{equation}
This representation allows us to interpret $a \in \FF$ as a polynomial $\tilde{a}(X) \in \BB[X]$ of degree at most $e-1$. We consider $n$ parties $P_1, \dots, P_n$ indexed by $[n] := \{1, \dots, n\}$, where each party $P_i$ is assigned a distinct nonzero evaluation point $\alpha_i \in \BB \setminus \{0\}$, requiring $n < q$.

We assume parties communicate over pairwise private channels in the presence of a static, semi-honest adversary that corrupts a subset $T\subseteq[n]$ of at most $c\ge 1$ parties. Corrupted parties faithfully follow the protocol specification, but the adversary pools their internal states and received messages to infer private information. In particular, the designated reconstruction party $\partyking$ may belong to $T$. A protocol $\pi$ securely evaluates an arithmetic circuit $C$ over $\FF$ (where all parties learn the output) with perfect simulation security~\cite{asharov2017full} if, for every $T \subseteq [n]$ with $|T| \le c$, there exists a probabilistic polynomial-time simulator $\mathcal{S}$ such that for every input vector $\boldsymbol{x} = (x_1, \dots, x_n) \in \FF^n$,
\begin{equation}
  \mathcal{S}(T, \boldsymbol{x}_T, C(\boldsymbol{x})) \equiv \mathsf{View}^{\pi}_T(\boldsymbol{x}),
\end{equation}
where $\boldsymbol{x}_T = (x_i)_{i \in T}$, $\mathsf{View}^{\pi}_T(\boldsymbol{x})$ denotes the joint view of the parties in $T$ during an execution of $\pi$, and $\equiv$ denotes identical probability distributions. Throughout this paper, $\log$ denotes the base-2 logarithm, and communication complexity measures the total number of bits transmitted between distinct parties.

\subsection{Atomic Sharing}
To communicate using elements in $\BB$, we utilize the atomic sharing~\cite{cramer2007atomic}. This construction embeds an extension-field secret into the low-degree coefficients of a polynomial over $\BB$ and distributes its evaluations as in Shamir's secret sharing~\cite{shamir1979share}.

\begin{definition}[Atomic sharing~\cite{cramer2007atomic}]
For $d \ge 2e-1$, a $d$-atomic sharing of $a \in \FF$ is the vector $[a; d] = (f_a(\alpha_i))_{i=1}^n \in \BB^n$, where party $P_i$ holds the share $[a; d]_i = f_a(\alpha_i)$ of the polynomial
\begin{equation}
  f_a(X) = \tilde{a}(X) + X^{2e-1} z_a(X), \qquad \deg z_a \le d - 2e + 1,
\end{equation}
whose noise coefficients in $z_a(X)$ are chosen independently and uniformly at random from $\BB$. To reconstruct $a$, one interpolates $f_a(X)$ from any $d+1$ shares and extracts its first $e$ coefficients.
\end{definition}

By construction, the coefficients of $f_a(X)$ of degrees $e, \dots, 2e-2$ are zero. Note that $z_a(X)$ contains $d - 2e + 2$ independent uniform coefficients and the evaluation points $\alpha_i$ are distinct and nonzero. Thus, any $d - 2e + 2$ shares leak no information about $a$, while any $d + 1$ shares suffice to reconstruct $a$. That is, $d$-atomic sharing is a $(d - 2e + 2)$-private and $(d + 1)$-correct ramp secret sharing scheme~\cite{blakley1984security}.

Atomic sharing is $\BB$-linear: given $[a; d]_i$ and $[b; d]_i$, party $P_i$ can locally add them to obtain a valid atomic share of $a+b$. In contrast, locally multiplying two atomic shares yields an evaluation of the product polynomial
\begin{equation}
  f_a(X) f_b(X) = \tilde{a}(X)\tilde{b}(X) + X^{2e-1} h(X)
\end{equation}
of degree at most $2d$, where the first $2e-1$ coefficients form the unreduced product $\tilde{a}(X)\tilde{b}(X)$ of degree at most $2e-2$. The local products do not form a valid $d$-atomic sharing $[ab; d]$ for subsequent circuit gates since the polynomial degree has doubled to $2d$, and the gap degrees $e, \dots, 2e-2$ are no longer zero. Converting these local products back into a $d$-atomic sharing $[ab; d]$ therefore requires a degree- and modulo-reduction step.

\section{Compact Sharing}
\label{sec:compact}
In this section, we formally define \emph{compact sharing} and discuss its application in our MPC protocol. Compact sharing is introduced to resolve the information leakage that occurs when one tries to combine atomic sharing directly with the preprocessing approach of Damg{\aa}rd--Nielsen~\cite{damgaard2007scalable}. Whenever degree-$2d$ sharings are used below, assume $n\ge2d+1$.

\subsection{Why an Atomic Mask Is Insufficient}
In the usual preprocessing approach, we mask a product with a random share before sending it to $\partyking$. However, if we mask using a random atomic sharing $[r;2d]$, then the corresponding atomic polynomial has zero coefficients at degrees $e,e+1,\ldots,2e-2$. Hence, when $\partyking$ reconstructs $f_a(X)f_b(X)+f_r(X)$, the coefficients in these gap degrees remain unmasked. This leaks partial information about the operands. For example, let $e=2$, with $\tilde{a}(X) =a_0+a_1X$ and $\tilde{b}(X)=b_0+b_1X$. The low-degree part of the masked polynomial is
\begin{equation}
 (a_0b_0+r_0)+(a_0b_1+a_1b_0+r_1)X+a_1b_1X^2.
\end{equation}
The coefficient $a_1b_1$ is exposed. Masking only the reduced field element is insufficient because $\partyking$ sees the entire interpolated polynomial.

\subsection{Construction and Reconstruction of Compact Sharing}
Compact sharing is designed to mask exactly these coefficients. Its low-degree part is a randomized non-reduced representative of the secret.
\begin{definition}[Compact sharing]\label{def:compact}
For $D\ge2e-1$, a $D$-compact sharing of $a\in\FF$ is $\langle a;D\rangle=(g_a(\alpha_i))_{i=1}^n$, where
\begin{equation}\label{eq:compact}
 g_a(X)=\widetilde a(X)+F(X)\rho_a(X)+X^{2e-1}Z_a(X).
\end{equation}
The polynomials $\rho_a$ and $Z_a$ have degrees at most $e-2$ and $D-2e+1$, respectively, and independently uniform coefficients in $\BB$. To reconstruct, interpolate $g_a$ from any $D+1$ shares, retain its coefficients of degrees $0,\ldots,2e-2$, and reduce the resulting polynomial modulo $F$.
\end{definition}
The key property of $g_a(X)$ is that its first $2e-1$ coefficients form a randomized polynomial $\tilde{a}(X) + F(X)\rho_a(X)$ of degree at most $2e-2$ that reduces to $\tilde{a}(X)$ modulo $F(X)$, while the high-degree noise $X^{2e-1} Z_a(X)$ is discarded prior to modulo reduction.

\begin{lemma}\label{lem:compact}
$D$-compact sharing is $(D - 2e + 2)$-private and $(D + 1)$-correct. Moreover, if $r$ is chosen uniformly at random from $\FF$, then its compact sharing polynomial $g_r(X)$ is uniformly distributed over the space $\BB[X]_{\le D}$ of all polynomials of degree at most $D$ over $\BB$.
\end{lemma}

\begin{IEEEproof}
Since $\deg g_a \le D$, any $D+1$ shares uniquely determine $g_a(X)$ via Lagrange interpolation, and reducing its low-degree part modulo $F(X)$ yields $\tilde{a}(X) + F(X)\rho_a(X) \equiv \tilde{a}(X) \pmod{F(X)}$, establishing $(D+1)$-correctness.

For privacy, fix $a \in \FF$ and $\rho_a(X)$, and let $t = D - 2e + 2$. For any $t$ distinct nonzero evaluation points $\beta_1, \dots, \beta_t \in \BB \setminus \{0\}$, evaluating the noise term $X^{2e-1} Z_a(X)$ corresponds to multiplying the $t$ independent uniform coefficients of $Z_a(X)$ by the matrix $\mathrm{diag}(\beta_1^{2e-1}, \dots, \beta_t^{2e-1})(\beta_i^{j-1})_{i,j=1}^t$. Since both the diagonal matrix and the Vandermonde matrix are invertible, the noise evaluations are uniformly distributed over $\BB^t$, establishing $(D - 2e + 2)$-privacy.

Finally, by Euclidean polynomial division by the monic degree-$e$ polynomial $F(X)$, every polynomial in $\BB[X]_{\le 2e-2}$ can be uniquely written as $F(X)\rho_r(X) + \tilde{r}(X)$ with quotient $\deg \rho_r \le e-2$ and remainder $\deg \tilde{r} \le e-1$. Thus, the map $(\tilde{r}, \rho_r) \mapsto \tilde{r}(X) + F(X)\rho_r(X)$ is a bijection onto $\BB[X]_{\le 2e-2}$. When $r \in \FF$ is uniform, $\tilde{r}(X)$ and $\rho_r(X)$ are independent and uniform, making the low-degree part uniform over $\BB[X]_{\le 2e-2}$. Combined with the independent uniform coefficients of $Z_r(X)$ in degrees $2e-1, \dots, D$, the full polynomial $g_r(X)$ is uniform over $\BB[X]_{\le D}$.
\end{IEEEproof}

By Lemma~\ref{lem:compact}, masking a degree-$2d$ product polynomial $f_a(X)f_b(X)$ with a random $2d$-compact polynomial $g_r(X)$ randomizes all coefficients up to degree $2d$, while shifting the underlying secret in $\FF$ by $r$. Consequently, our preprocessing phase generates matched pairs $([r; d], \langle r; 2d \rangle)$ consisting of a $d$-atomic sharing and a $2d$-compact sharing of the same random element $r \in \FF$, using independent noise polynomials conditioned on $r$.

To illustrate how a matched pair $([r; d], \langle r; 2d \rangle)$ enables degree and modulo reduction for a multiplication gate, suppose each party $P_i$ holds $[a; d]_i$ and $[b; d]_i$:
\begin{enumerate}
  \item Party $P_i$ locally multiplies its atomic shares, masks the product with its compact share, and sends $[a; d]_i [b; d]_i + \langle r; 2d \rangle_i \in \BB$ to $\partyking$.
  \item Party $\partyking$ interpolates $f_a(X)f_b(X) + g_r(X)$, retains the first $2e-1$ coefficients, and reduces modulo $F(X)$ to recover the masked product $ab + r \in \FF$.
  \item Party $\partyking$ samples a fresh atomic sharing $[ab + r; d]$ and sends $[ab + r; d]_i \in \BB$ to each party $P_i$.
  \item Party $P_i$ locally subtracts $[r; d]_i$ to obtain the valid atomic share $[ab; d]_i = [ab + r; d]_i - [r; d]_i$.
\end{enumerate}
Here, the compact sharing $\langle r; 2d \rangle$ masks the unreduced product polynomial opened to $\partyking$, while the companion atomic sharing $[r; d]$ allows the parties to locally strip the mask $r$ after degree reduction. Section~\ref{sec:mpc} extends this procedure to general multiplication gates whose inputs are affine combinations over $\FF$.
\section{Small Shares MPC from Atomic-Compact Sharing}\label{sec:mpc}
Following the preprocessing paradigm of~\cite{damgaard2007scalable}, our MPC protocol consists of two phases: a preprocessing phase that generates batches of random matched atomic-compact sharing pairs, and an online evaluation phase that evaluates the arithmetic circuit over secret shares.

\subsection{Circuit Representation and Parameters}
While atomic sharing is homomorphic over $\BB$, multiplying a shared wire $a \in \FF$ by a public scalar $s \in \FF \setminus \BB$ corresponds to multiplying $f_a(X)$ by $\tilde{s}(X)$ of degree at most $e-1$. This increases the polynomial degree from $d$ to $d + e - 1$ and populates the gap degrees $e, \dots, 2e-2$, requiring a degree- and modulo-reduction step just like multiplication of two secrets.

To avoid performing a separate reduction for every scalar multiplication, we view the arithmetic circuit $C$ over $\FF$ as a straight-line program with fan-in-two multiplication gates and absorb all intermediate affine operations over $\FF$ into the inputs of the next multiplication gate. Thus, each non-linear multiplication gate computes a product $AB$, where
\begin{equation}\label{eq:mult-gate}
  A = s + \sum_{j \in S} s_j a_j, \qquad B = t + \sum_{k \in R} t_k b_k,
\end{equation}
all scalars $s, s_j, t, t_k$ lie in $\FF$, and $a_j, b_k$ are circuit inputs or outputs of prior multiplication gates. Any trailing affine combination at the output of the circuit is evaluated as a multiplication gate~\eqref{eq:mult-gate} with $B = 1$; thus, the circuit consists of $\mathcal{C}$ generalized multiplication gates of the form~\eqref{eq:mult-gate}, and every circuit output wire is the output of a multiplication gate.

To tolerate $c$ semi-honest corruptions while enabling $\partyking$ to reconstruct degree-$2d$ polynomials from $n$ shares, we set
\begin{equation}\label{eq:params}
  d = c + 2e - 2, \qquad D = 2d, \qquad 2c + 4e - 3 \le n < q.
\end{equation}
Under~\eqref{eq:params}, the privacy threshold of $d$-atomic sharing is $d - 2e + 2 = c$, and the reconstruction threshold of $2d$-compact sharing satisfies $2d + 1 = 2c + 4e - 3 \le n$. These conditions capture the tradeoff between the extension degree $e$ and the corruption threshold $c$.

\subsection{Preprocessing Phase}
In the preprocessing phase, all parties jointly generate batches of random matched pairs $([r_k; d], \langle r_k; 2d \rangle)$. Set $m = n - c$ and let $M \in \BB^{m \times n}$ be the Vandermonde matrix with entries $M_{ki} = \alpha_i^{k-1}$. In each batch, each party $P_i$ samples $\sigma_i \in \FF$ uniformly at random, independently generates $[\sigma_i; d]$ and $\langle \sigma_i; 2d \rangle$ conditioned on $\sigma_i$, and sends the pair $([\sigma_i; d]_j, \langle \sigma_i; 2d \rangle_j) \in \BB^2$ to each party $P_j$. Each party $P_j$ then locally computes, for $1 \le k \le m$,
\begin{equation}\label{eq:preproc}
  \hspace{-2mm}[r_k; d]_j = \sum_{i=1}^n M_{ki} [\sigma_i; d]_j, \,\,
  \langle r_k; 2d \rangle_j = \sum_{i=1}^n M_{ki} \langle \sigma_i; 2d \rangle_j.
\end{equation}
Because the Vandermonde entries $M_{ki}$ lie in the base field $\BB$ and both sharing schemes are $\BB$-linear, \eqref{eq:preproc} yields valid matched pairs for $r_k = \sum_{i=1}^n M_{ki} \sigma_i \in \FF$ without increasing the polynomial degrees. Since any $m$ columns of $M$ form an invertible Vandermonde matrix over $\BB$, the $m = n - c$ honest parties' contributions ensure that the extracted masks $r_1, \dots, r_m$ are independent and uniform even when conditioned on the $c$ corrupted parties' contributions (proved formally in Section~\ref{sec:security}). Each batch communicates $2n(n-1)$ elements of $\BB$ and yields $m = n - c$ matched pairs. Since each multiplication gate consumes three fresh pairs (see Section~\ref{subsec:eval}), the parties execute $\lceil 3\mathcal{C}/(n-c) \rceil$ batches in total.

\subsection{Evaluation Phase}\label{subsec:eval}
Once the required preprocessing pairs are prepared, the parties evaluate the circuit gate by gate as follows.

\paragraph{Input gates}
For each circuit input $x \in \FF$ held by a party $P_i$, party $P_i$ samples a $d$-atomic sharing $[x; d]$ and sends $[x; d]_j \in \BB$ to each other party $P_j$.

\paragraph{Multiplication gates}
To evaluate a multiplication gate $AB$ as in~\eqref{eq:mult-gate} using three preprocessed pairs $([r_\ell; d], \langle r_\ell; 2d \rangle)$ for $\ell = 1, 2, 3$:
\begin{enumerate}
  \item Each party $P_i$ locally computes the masked evaluations
  \begin{align*}
    u_i &= \langle r_1; 2d \rangle_i + \tilde{s}(\alpha_i) + \sum_{j \in S} \tilde{s}_j(\alpha_i) [a_j; d]_i, \\
    v_i &= \langle r_2; 2d \rangle_i + \tilde{t}(\alpha_i) + \sum_{k \in R} \tilde{t}_k(\alpha_i) [b_k; d]_i,
  \end{align*}
  and sends $u_i, v_i \in \BB$ to $\partyking$.
  \item Party $\partyking$ applies $2d$-compact reconstruction to $(u_i)_{i=1}^n$ and $(v_i)_{i=1}^n$ to recover $U = A + r_1 \in \FF$ and $V = B + r_2 \in \FF$, independently samples fresh $d$-atomic sharings $[U; d]$ and $[V; d]$, and sends $[U; d]_i, [V; d]_i \in \BB$ to each party $P_i$.
  \item Each party $P_i$ locally unmasks the affine combinations and multiplies them to obtain
  \begin{align*}
    a'_i &= [U; d]_i - [r_1; d]_i, \qquad b'_i = [V; d]_i - [r_2; d]_i, \\
    w_i &= a'_i b'_i + \langle r_3; 2d \rangle_i,
  \end{align*}
  and sends $w_i \in \BB$ to $\partyking$.
  \item Party $\partyking$ applies $2d$-compact reconstruction to $(w_i)_{i=1}^n$ to recover $W = AB + r_3 \in \FF$, samples a fresh $d$-atomic sharing $[W; d]$, and sends $[W; d]_i \in \BB$ to each party $P_i$. Party $P_i$ locally outputs $[AB; d]_i = [W; d]_i - [r_3; d]_i$.
\end{enumerate}
Crucially, $\partyking$ returns base-field shares $[U; d]_i, [V; d]_i, [W; d]_i \in \BB$ rather than broadcasting extension-field elements in $\FF$. Steps~1--2 convert the two $\FF$-affine combinations $A$ and $B$ in parallel into $d$-atomic sharings, and Steps~3--4 evaluate their product.

\paragraph{Output gates}
Given the atomic sharing $[y; d]$ of an output wire, each party $P_i$ sends $[y; d]_i \in \BB$ to $\partyking$, who reconstructs $y \in \FF$ and sends $y$ to every other party.

\subsection{Correctness of the MPC Protocol}
Correctness of the input and output gates follows directly from the $(d+1)$-correctness of atomic sharing since $n \ge 2d + 1 > d + 1$. For each $j \in S$ in a multiplication gate, the atomic sharing $[a_j; d]$ corresponds to $f_{a_j}(X) = \tilde{a}_j(X) + X^{2e-1} z_{a_j}(X)$. Multiplying its evaluation by the public evaluation $\tilde{s}_j(\alpha_i)$ and summing alongside $\tilde{s}(\alpha_i)$ evaluates a polynomial of degree at most $d + e - 1 \le 2d$. Since all noise terms remain divisible by $X^{2e-1}$, its coefficients of degrees $0, \dots, 2e-2$ equal $\tilde{s}(X) + \sum_{j \in S} \tilde{s}_j(X)\tilde{a}_j(X)$, which reduces modulo $F(X)$ to $\tilde{A}(X)$.

By adding the compact share $\langle r_1; 2d \rangle_i$, the parties mask this polynomial, and compact reconstruction modulo $F(X)$ yields $U = A + r_1$. By identical logic, $\partyking$ reconstructs $V = B + r_2$. Because atomic sharing is $\BB$-linear, each party locally subtracts its preprocessed atomic shares to obtain valid atomic shares:
\begin{equation*}
  [U; d]_i - [r_1; d]_i = [A; d]_i, \qquad [V; d]_i - [r_2; d]_i = [B; d]_i.
\end{equation*}
The local product of these shares evaluates a degree-$2d$ polynomial whose coefficients in degrees $0, \dots, 2e-2$ equal $\tilde{A}(X)\tilde{B}(X)$. Adding $\langle r_3; 2d \rangle_i$ allows $\partyking$ to recover $W = AB + r_3$, and $[W; d] - [r_3; d]$ yields a valid atomic sharing of $AB$. Induction over the multiplication gates establishes correctness of the circuit evaluation.

\section{Security and Communication Analysis}\label{sec:security}

\subsection{Perfect Privacy of the Joint View}
To prove perfect privacy against a corrupted coalition $T \subseteq [n]$ with $|T| \le c$ (which may include $\partyking$), we must show that the adversary's joint view---comprising its preprocessing messages, internal randomness, all polynomials opened to $\partyking$ during multiplication gates, and the output shares---can be perfectly simulated from $\boldsymbol{x}_T$ and the circuit output alone. We establish this in three steps: (1) characterizing the conditional distribution of preprocessed pairs, (2) showing that online conversions leak no information via a canonical linear lift $L_T$, and (3) constructing the full simulator $\mathcal{S}$.

\paragraph{Conditional privacy of preprocessing}
Let $E_T$ denote the map evaluating a polynomial at the corrupted points $\{\alpha_i : i \in T\}$, and define the $\BB$-linear decoding map on $\BB[X]_{\le 2d}$ by
\begin{equation*}
  \mathrm{Dec}(p) = (p \bmod X^{2e-1}) \bmod F(X) \in \FF.
\end{equation*}
For a fresh matched pair $(f_r, g_r)$ with uniform $r \in \FF$, Lemma~\ref{lem:compact} implies that $g_r$ is uniform over $\BB[X]_{\le 2d}$. Moreover, for any fixed $g_r$ (which fixes $r = \mathrm{Dec}(g_r)$), the independent noise in $f_r$ makes $E_T f_r$ uniform over $\BB^{|T|}$. Thus, conditioned on the corrupted parties' shares $(E_T f_r, E_T g_r)$, the pair $(f_r, g_r)$ is uniformly distributed over the affine space of matched pairs $(\hat{f}, \hat{g})$ satisfying $\mathrm{Dec}(\hat{f}) = \mathrm{Dec}(\hat{g})$ and $(E_T \hat{f}, E_T \hat{g}) = (E_T f_r, E_T g_r)$.

This conditional distribution is preserved by Vandermonde batch extraction~\eqref{eq:preproc}. Condition on all preprocessing shares sent to $T$ and all corrupted dealers' randomness. Each honest dealer's pair is uniform on a coset of the joint evaluation kernel $\mathcal{K}_T = \{(\hat{f}, \hat{g}) : \mathrm{Dec}(\hat{f}) = \mathrm{Dec}(\hat{g}),\, E_T \hat{f} = E_T \hat{g} = 0\}$. Selecting any $m = n - c$ honest dealers and fixing the remaining dealers' contributions, the corresponding invertible $m \times m$ Vandermonde submatrix of $M$ maps $\mathcal{K}_T^m$ bijectively onto itself. Thus, the $m$ extracted pairs in each batch are conditionally independent and follow the exact same conditional distribution as fresh matched pairs.

\paragraph{Online gate privacy}
To relate real and simulated wire polynomials during online evaluation, let $u = |T| \le c$ and define the explicit $\BB$-linear map $L_T : \FF \to \BB[X]_{\le d}$ for any $\Delta \in \FF$ by
\begin{equation}\label{eq:lift}
  L_T(\Delta)(X) = \tilde{\Delta}(X) + X^{2e-1} \sum_{j=0}^{u-1} z_j(\Delta) X^j,
\end{equation}
where $(z_0(\Delta), \dots, z_{u-1}(\Delta)) \in \BB^u$ is the unique solution to $L_T(\Delta)(\alpha_i) = 0$ for all $i \in T$ (with the sum empty if $u = 0$). Since the evaluation points $\{\alpha_i : i \in T\}$ are distinct and nonzero, the existence and uniqueness of $(z_j(\Delta))_{j=0}^{u-1}$ are guaranteed by the invertibility of the corresponding Vandermonde matrix. By construction, $L_T(\Delta)$ is a valid $d$-atomic polynomial satisfying $\mathrm{Dec}(L_T(\Delta)) = \Delta$ and $E_T L_T(\Delta) = 0$.

During online evaluation, each multiplication gate performs three degree-reduction conversions: two for the affine inputs in Step~1 (where the unmasked polynomial is $\tilde{s}(X) + \sum_{j \in S} \tilde{s}_j(X)f_{a_j}(X)$ or its $B$-analog) and one for the product in Step~3 (where the unmasked polynomial is $f_A(X)f_B(X)$). Fix any such conversion step, and let $p \in \BB[X]_{\le 2d}$ and $p' \in \BB[X]_{\le 2d}$ denote the unmasked polynomials formed in the real and simulated executions, respectively. Suppose by induction that on every wire entering this step, the real atomic polynomial $f_a$ and simulated atomic polynomial $f_{a'}$ satisfy
\begin{equation}\label{eq:wire-invariant}
  f_a(X) - f_{a'}(X) = L_T(a - a')(X).
\end{equation}
Since $E_T L_T(a - a') = 0$, the corrupted parties hold identical wire shares ($E_T f_a = E_T f_{a'}$), and because each party computes its unmasked evaluation locally from its wire shares, the unmasked polynomials agree on $T$ ($E_T p = E_T p'$). Consequently, the affine map on the matched-pair space given by
\begin{equation}\label{eq:coupling}
  (f_r, g_r) \longmapsto (\tilde{f}_r, \tilde{g}_r) := \big(f_r + L_T(\mathrm{Dec}(p - p')),\, g_r + p - p'\big)
\end{equation}
is a bijection on the conditional coset of valid matched pairs: since $E_T L_T(\mathrm{Dec}(p - p')) = 0$ and $E_T(p - p') = 0$, it preserves the corrupted shares $E_T \tilde{f}_r = E_T f_r$ and $E_T \tilde{g}_r = E_T g_r$ as well as the matched-secret condition $\mathrm{Dec}(\tilde{f}_r) = \mathrm{Dec}(\tilde{g}_r)$. Under this coupling, the polynomial opened to $\partyking$ in the simulated execution is $p' + \tilde{g}_r = p + g_r$, which is identical to the real opening and thus leaks no information about $p$ beyond $E_T p$. Moreover, holding $\partyking$'s resharing polynomial $f_U$ fixed, the resulting unmasked wire polynomials $f_s = f_U - f_r$ (for the real wire value $s = \mathrm{Dec}(p)$) and $f_{s'} = f_U - \tilde{f}_r$ (for the simulated wire value $s' = \mathrm{Dec}(p')$) satisfy $f_s - f_{s'} = \tilde{f}_r - f_r = L_T(s - s')$, preserving the invariant~\eqref{eq:wire-invariant} at the output of the conversion.

\paragraph{Full circuit simulation}
The simulator $\mathcal{S}$ uses the corrupted inputs $\boldsymbol{x}_T$, sets all honest inputs to $0 \in \FF$, and executes the preprocessing, input, and multiplication gates honestly. Coupling each honest party's real input polynomial $f_{x_i}$ with its simulated input polynomial $f_0$ via $f_{x_i} = f_0 + L_T(x_i)$ establishes the invariant~\eqref{eq:wire-invariant} at the input wires (with $E_T f_{x_i} = E_T f_0$). Applying~\eqref{eq:coupling} inductively across both affine conversions and the product conversion of every multiplication gate preserves~\eqref{eq:wire-invariant} throughout the circuit while producing an identically distributed online transcript. Finally, for each output wire with true value $y$ and simulated value $y'$, \eqref{eq:wire-invariant} guarantees that the real and simulated output polynomials satisfy $f_y(X) - f_{y'}(X) = L_T(y - y')(X)$. Before opening the output wire to $\partyking$, $\mathcal{S}$ adds $L_T(y - y')(\alpha_i)$ to each honest party's simulated output share $[y'; d]_i$ (leaving corrupted shares unchanged since $E_T L_T(y - y') = 0$). This maps the reconstructed polynomial $f_{y'}(X) + L_T(y - y')(X)$ to the exact conditional distribution of $f_y(X)$, completing the simulation.

\subsection{Exact Communication and Amortization}
Let $\ell_I$ be the number of circuit inputs and $\ell_O$ be the number of circuit outputs in $\FF$. Evaluating $\mathcal{C}$ multiplication gates requires $3\mathcal{C}$ matched pairs, which are generated using $b = \lceil 3\mathcal{C}/(n-c) \rceil$ preprocessing batches. Using $\log q$ bits per base-field element and $\log|\FF| = e \log q$ bits per extension-field element (or $\lceil \log q \rceil$ and $\lceil \log|\FF| \rceil$ under fixed-length binary encodings), the protocol communicates
\begin{equation}\label{eq:comm-exact}
\begin{aligned}
  B_{\mathrm{pre}} &= 2n(n-1)b \log q, \\
  B_{\mathrm{eval}} &= (n-1)(\ell_I + 6\mathcal{C} + \ell_O)\log q + \ell_O(n-1)\log|\FF|.
\end{aligned}
\end{equation}
Here, each input gate communicates $n-1$ base-field elements, each multiplication gate performs three conversions costing $2(n-1)$ base-field elements each, and each output gate communicates $n-1$ base-field elements to $\partyking$ followed by $n-1$ extension-field elements broadcast by $\partyking$.

\begin{theorem}\label{thm:main}
Under~\eqref{eq:params}, the protocol correctly evaluates the circuit with perfect security against at most $c$ static semi-honest corruptions. Excluding input and output gates, the amortized communication cost per multiplication gate over full preprocessing batches is
\begin{equation}\label{eq:comm-amortized}
  \left(6(n-1) + \frac{6n(n-1)}{n-c}\right)\log q
  = O\left(\frac{n^2}{n-c}\,\frac{\log|\FF|}{e}\right) \text{ bits.}
\end{equation}
For finite circuits, the exact communication cost is given by~\eqref{eq:comm-exact}, where rounding up $b$ adds at most one incomplete preprocessing batch of $2n(n-1)\log q$ bits.
\end{theorem}

When the extension degree scales linearly with the number of parties, i.e., $e = \lambda n$ for a constant $0 < \lambda < 1/4$, condition~\eqref{eq:params} yields $c \le n(1/2 - 2\lambda) + 3/2$ and $n - c = \Theta(n)$. In this regime, $n^2/(n-c) = O(n) = O(e)$, so~\eqref{eq:comm-amortized} simplifies to $O(\log|\FF|)$ bits per multiplication gate, achieving constant amortized communication overhead relative to the plaintext size.

\clearpage
\balance

\end{document}